\documentclass[conference]{IEEEtran}
\IEEEoverridecommandlockouts
\usepackage{amsthm}
\usepackage{caption}
\usepackage{subcaption}
\usepackage{cite}
\usepackage{amsmath,amssymb,amsfonts}
\usepackage{algorithmic}
\usepackage{graphicx}
\usepackage{textcomp}
\usepackage{xcolor}
\usepackage{comment}
\usepackage{subcaption}
\def\BibTeX{{\rm B\kern-.05em{\sc i\kern-.025em b}\kern-.08em
    T\kern-.1667em\lower.7ex\hbox{E}\kern-.125emX}}

\newtheorem{theorem}{Theorem}

\newtheorem{lemma}[theorem]{Lemma}

\usepackage{enumitem}

\usepackage{color, soul}

\newcommand{\mA}{m^{A\backslash B}}
\newcommand{\mAB}{m^{AB}}
\newcommand{\mB}{m^{B\backslash A}}
\newcommand{\A}{A\backslash B}
\newcommand{\AB}{AB}
\newcommand{\B}{B\backslash A}

\begin{document}

\title{Impact of Geographical Separation on Spectrum Sharing Markets
}
\author{
\IEEEauthorblockN{Kangle Mu, Zongyun Xie, Igor Kadota, and Randall Berry}
\IEEEauthorblockA{Department of Electrical and Computer Engineering\\
    Northwestern University\\
    \{kangle.mu, zongyun.xie, kadota, rberry\}@northwestern.edu}\\%
}

\maketitle

\begin{abstract}
With the increasing demand for wireless services, spectrum management agencies and service providers (SPs) are seeking more flexible mechanisms for spectrum sharing to accommodate this growth. Such mechanisms impact the market dynamics of competitive SPs. Prior market models of spectrum sharing largely focus on scenarios where competing SPs had identical coverage areas.  We depart from this and consider a scenario in which two competing SPs have overlapping but distinct coverage areas. We study the resulting competition using a Cournot model.
 Our findings reveal that with limited shared bandwidth, SPs might avoid overlapping areas to prevent potential losses due to interference. Sometimes SPs can strategically cooperate by agreeing not to provide service in the overlapping areas and, surprisingly, customers might also benefit from such cooperation under certain circumstances. Overall, market outcomes exhibit complex behaviors that are influenced by the sizes of coverage areas and the bandwidth of the shared spectrum.
\end{abstract}


\section{Introduction}
Spectrum sharing is receiving increased interest to meet the ever growing demands for wireless services. Examples include the recent U.S. National Spectrum Strategy \cite{nationalspectrum} and programs such as the Spectrum Innovation Initiative: National Radio Dynamic Zones (SII-NRDZ) that seeks to further advance dynamic spectrum sharing \cite{nrdz}. 

Spectrum sharing involves both technical and economic dimensions in that it impacts how service providers (SPs) compete with each other. 
Models of competition with various forms of shared spectrum have been studied including 
\cite{nguyen2016cost,maille2012competition,berry2017value, mu2024pooling,mu2024relaxation,xie2023market}.  In these works, it was assumed that all competing SPs had the same coverage area, so that any customer could be served by any SP.  In this paper we depart from this and consider an example where two competing SPs are sharing the same spectrum and have distinct, partially overlapping coverage areas.  For example, this could model two WiFi providers sharing the same band of unlicensed spectrum, but with different coverage due to the placement of their access points (APs).  We seek to understand the impact of the geographic separation between the SPs on their competition.

We consider a scenario in which two competing SPs each have an AP at a distinct location with overlapping coverage. We assume that the same spectrum band is used by two SPs.  We categorize the coverage into two types: dedicated areas served exclusively by one AP, and an overlapping area served by both APs. 
The SPs compete for a pool of customers spread across these areas.   The users are congestion-sensitive in that the price they are willing to pay depends on a congestion cost, which in turn varies across these areas, modeling different levels of interference that may occur. 
As in \cite{berry2017value}, we adopt a Cournot competition model in which both SPs specify the number of customers they want to serve in both the dedicated and overlapping sub-markets. 
Our main results are as follows:
\begin{itemize}[leftmargin=0.15in]
  \item We prove that a unique Nash equilibrium always exists in the proposed model.
  \item With limited bandwidth, SPs typically avoid overlapping areas to minimize the risk of significant congestion, which could adversely affect their revenue. However, with sufficient bandwidth, SPs will enter the overlapping sub-market. If this happens, the SPs may incur revenue losses due to the competition in the overlapping market, and consumer surplus may be redcued.
  \item SPs might want to cooperate by agreeing not to serve users in the overlapping areas to avoid competition. Surprisingly, sometimes consumers can also benefit from such cooperation in the sense of total consumer surplus, which in turn leads to higher social welfare. However, this may also raise concerns about fairness, as no customers are served in the overlapping areas.
  \item Market outcomes, including consumer surplus and social welfare, exhibit a complex dynamic and may not necessarily increase with the bandwidth provided to the SPs. This suggests that regulators need to carefully determine the amount of shared spectrum to optimize these outcomes.
\end{itemize}

Regarding related work, we follow the stream of modeling wireless spectrum as congestible resources \cite{nguyen2016cost, berry2017value, mu2024pooling, mu2024relaxation} and our analysis builds on the framework in \cite{berry2017value} where a market with intermittent spectrum is considered. Here, we instead consider a non-intermittent band of spectrum and account for the geographical differences in the locations of SPs. Our work is also related to work on access point or base station placement, e.g.~\cite{altman2012spatial, al2020UAV, salameh2022intelligent, coluccia2012sinr}, but these studies focus more on the technical aspects of spectrum usage rather than its economic impacts.  This work also has ties to approaches that SPs could use to co-ordinate their spectrum usage, such as the use of Spectrum Consumption Models (SCM)~\cite{caicedo2015SCM,bastidas2018ieee,Large-scale,SCM_COSMOS}.

\section{Market Model}
We consider a model in which there are 2 SPs (SP1 and SP2), each deploying an AP in distinct locations but with overlapping coverage. Both SPs compete for a common pool of non-atomic customers, who are categorized into three groups according to the coverage area they fall under. Let $A$ and $B$ be the sets of users under the coverage of SP1 and SP2, respectively. Then, we have three sub-markets with the corresponding sets of customers denoted by $A\backslash B$, $AB = A\cap B$, and $B\backslash A$, and with the market sizes $\mA$, $\mAB$, and $\mB$, respectively. For ease of analysis, we assume customers are non-atomic and, without loss of generality, the total market size $\mA+\mAB+\mB=1$.

In each sub-market, Cournot competition is considered in which SPs announce the quantity of users they want to serve, and this in turn leads to a market-clearing price \cite{berry2017value}. Let $x_1^{\A}$, $x_1^{\AB}$, $x_2^{\AB}$, and $x_2^{\B}$ denote the quantities of users served by the SPs indicated by their subscripts. Fig.~\ref{figure_venn} shows how users served by different SPs fall into different sub-markets. Since SPs can only serve users under their own coverage, we have the following constraints
\begin{IEEEeqnarray}{rl}
0 \leq& x_1^{\A}\leq \mA, \label{equ_x_dom1}\\
0 \leq& x_1^{\AB} + x_2^{\AB}\leq \mAB, \label{equ_x_dom2}\\
0 \leq& x_2^{\B}\leq \mB.\label{equ_x_dom3}
\end{IEEEeqnarray}

We define the market-clearing prices of each sub-market as follows
\begin{IEEEeqnarray}{ll}
p_d^{\A} &= 1 - \frac{x_1^{\A}}{\mA},\label{equ_pd_A}\\
p_d^{\AB} &= 1 - \frac{x_1^{\AB}+x_2^{\AB}}{\mAB},\label{equ_pd_AB}\\
p_d^{\B} &= 1 - \frac{x_2^{\B}}{\mB}.\label{equ_pd_B}
\end{IEEEeqnarray}

\begin{figure}[tb]
\centering
\includegraphics[width=5.2cm]{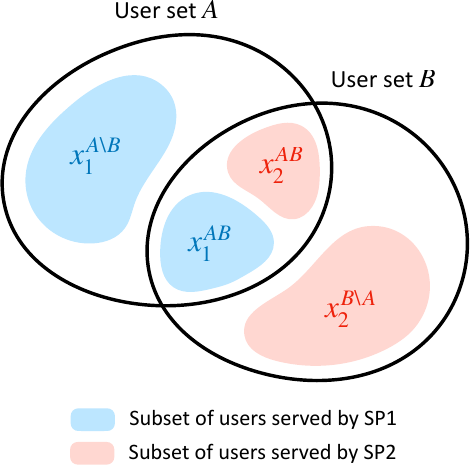}
\caption{Venn diagram of three sub-markets and the sets of users served by different SPs.}
\label{figure_venn}
\end{figure}

These definitions are based on the following assumptions:
\begin{enumerate}
\item We assume that the whole market has the following linear (inverse) demand function
\begin{IEEEeqnarray}{ll}
p_d&= 1 - x,\label{equ_pd}
\end{IEEEeqnarray}
where $x=x_1^{\A} + x_1^{\AB} + x_2^{\AB} + x_2^{\B}$ is the total quantity of users served in the whole market.
\item Customers, differing in their valuation (i.e., utility) of the wireless service, are assumed to be uniformly distributed among sub-markets. In other words, within each sub-market, there are all types of users ranging from high-value users who are willing to pay for the service with a higher price to low-value users who have a limited budget. The mass of different types of users is proportional to the sub-market sizes.
\end{enumerate}
One can verify that under definition (\ref{equ_pd_A})--(\ref{equ_pd_B}), the only way to get the same market-clearing price across three sub-markets is to serve users in proportion to the sizes of sub-market, and the resulting price is given by \eqref{equ_pd}.

Next, we derive the consumer surplus in each sub-market. Take market $\A$ as an example. Equation (\ref{equ_pd_A}) specifies the inverse demand in market $\A$, i.e., it indicates the minimum price at which a mass of  $x_1^{\A}$ customers would accept service.  It follows that the surplus of the $x$th user is given by $1-x/m^{\A}-p_d^{\A}$. To derive the consumer surplus of the entire sub-market, we need to integrate the surplus over $x$ from $0$ to $x_1^{\A}$. This results in the consumer surplus of each sub-market given by
\begin{IEEEeqnarray}{ll}
CS^{\A} &=  \frac{(x_1^{\A})^2}{2\mA},\label{eq_cs_A}\\
CS^{\AB} &=  \frac{(x_1^{\AB}+x_2^{\AB})^2}{2\mAB},\label{equ_cs_AB}\\
CS^{\B} &= \frac{(x_2^{\B})^2}{2\mB}.\label{equ_cs_B}
\end{IEEEeqnarray}

We refer to the market clearing prices in (\ref{equ_pd_A})--(\ref{equ_pd_B}) as {\it delivered prices}. We assume that users are sensitive to congestion that is measured in terms of {\it latency costs}. Then the {\it service price} charged by an SP to its users is given by the difference between the delivered price and the latency cost. This models users that may avoid low-cost but poor-quality wireless services. A user will use the service of a SP only if the sum of both costs is lower than their valuation of this service. 

We assume the latency cost incurred by users on a band in a given sub-market depends on the total number of users impacting that sub-market divided by its bandwidth, where as described below a user may impact (e.g. cause interference) a sub-market even if it is not present in that sub-market.
 If both SPs use the same band, the traffic from both should be considered in this calculation. Assume a band with bandwidth $W$ is used by both SPs. Considering their different coverage, we define the latency costs for the three sub-markets as follows
\begin{IEEEeqnarray}{ll}
l^{\A}&= \frac{x_1^{\A}+x_1^{\AB}+x_2^{\AB}}{W},\label{equ_letency_A}\\
l^{\AB}&= \frac{x_1^{\A}+x_1^{\AB}+x_2^{\AB}+x_2^{\B}}{W},\label{equ_letency_AB}\\
l^{\B}&= \frac{x_1^{\AB}+x_2^{\AB}+x_2^{\B}}{W}.\label{equ_letency_B}
\end{IEEEeqnarray}
We have a few comments on (\ref{equ_letency_A})--(\ref{equ_letency_B}):
\begin{enumerate}
    \item  Users within the same sub-market incur the same latency cost regardless of which AP they are connected to. To simplify the model, we assume the latency incurred by users mainly comes from the congestion caused by user traffic as opposed to path loss or shadowing.
    \item The latency model here is motivated by WiFi in which carrier-sense multiple access with collision avoidance (CSMA/CA) is used for multiple access. Namely, we assume that at most one user within the range of a given AP can transmit at a time (or be transmitted to by the AP). For example, the latency of users in $\A$ will then depend on $x_1^{\AB}$,  $x_2^{\AB}$, and $x_1^{\A}$ as all of these users are within range of AP1.  Likewise, the latency of  users in $\AB$ will depend on the number of users within range of either AP as these users are within range of both APs.\footnote{This is  a  simplification of an actual WiFi setting made to capture the key feature that with different geographic coverage, latency costs will depend on users within and external to a sub-market.} 
    \item From (\ref{equ_letency_AB}) one can conclude that market $AB$ is always ``more crowded'' compared to market $\A$ and $\B$ as users there always incur higher latency costs. This may give SPs a preference for market $\A$ and $\B$. In the next section, we will see how this preference would change after factoring in market sizes. 
\end{enumerate}

The revenue of an SP is the product of its service price and the quantity of users it serves. Thus, the revenue of each SP is given by
\begin{IEEEeqnarray}{rl}
R_1 &= x_1^{\A} \left(p_d^{\A}- l^{\A}\right) + x_1^{\AB} \left(p_d^{\AB}- l^{\AB}\right),\label{equ_R1}\\
R_2 &= x_2^{\B} \left(p_d^{\B}- l^{\B}\right) + x_2^{\AB} \left(p_d^{\AB}- l^{\AB}\right).\label{equ_R2}
\end{IEEEeqnarray}
Note here we  are assuming that the SPs can differentiate their prices across sub-markets, which requires them to know which users are in which sub-market.  This could be learned through measurements or shared if the SPs coordinate with each other, e.g. by using SCMs \cite{caicedo2015SCM}.

The SPs' goal is to maximize their revenue (\ref{equ_R1}) and (\ref{equ_R2}) by carefully choosing the quantities of users to serve for each sub-market (i.e., $x_1^{\A}$, $x_1^{\AB}$, $x_2^{\AB}$, and $x_2^{\B}$). The revenues of SP1 and SP2 couple with each other's decision through both the delivered prices (\ref{equ_pd_A})--(\ref{equ_pd_B}) and the latency costs (\ref{equ_letency_A})--(\ref{equ_letency_B}). Such coupling makes it a game between SP1 and SP2. We will discuss the Nash equilibrium and the corresponding market outcomes in the next section.

\section{Main Results}
In this section, we first characterize the Nash equilibrium of this two-player game, and based on that, we will examine welfare measures such as consumer surplus and social welfare.
\subsection{Equilibrium}
\begin{theorem}[Uniqueness of Nash equilibrium]
\label{theorem_uniqueness}
There always exists a unique Nash equilibrium for any bandwidth $W$, and sub-market sizes $\mA$, $\mAB$, and $\mB$.

For symmetric cases in which $\mA = \mB$, the quantities of users served by SP1 and SP2 at the equilibrium are given as follows:

\begin{IEEEeqnarray}{ll}
    x_1^{\A}&=
    \begin{cases} 
    \frac{Wm^{\A}}{2(W+m^{\A})}, & 0 \leq W < \frac{m^{\A}}{2}\vspace{5pt}\\
      \frac{W m^{\A}}{C}, & W\geq \frac{m^{\A}}{2}
   \end{cases}\label{equ_x1}\\
   x_1^{\AB}&=
    \begin{cases} 
    0, & 0 \leq W < \frac{m^{\A}}{2}\vspace{5pt}\\
      \frac{(2W-m^{\A}) m^{\AB}}{3C}, & W\geq \frac{m^{\A}}{2}
   \end{cases}\label{equ_x2} \\
   x_2^{\B} &= x_1^{\A}, \mbox{ and }  x_2^{\AB} = x_1^{\AB}, \nonumber
\end{IEEEeqnarray}
where
\begin{IEEEeqnarray}{c}
C=2(W+m^{\A}+m^{\AB})-\frac{m^{\A}m^{\AB}}{W}.
\end{IEEEeqnarray}
\end{theorem}
\begin{proof}
Due to space considerations, we only provide the proof for the symmetric case, i.e., $\mA = \mB$. 

Before diving into the proof of uniqueness, we first provide some intuition, which is important for understanding how the proof is constructed. 

One can solve the following first-order conditions to gain insight into a potential equilibrium:
\begin{IEEEeqnarray}{c}
   \left[\frac{\partial R_1}{\partial x_1^{\A}}~\frac{\partial R_1}{\partial x_1^{\AB}} ~\frac{\partial R_2}{\partial x_2^{\AB}}~\frac{\partial R_2}{\partial x_2^{\B}}\right]^T
    =\mathbf{0}.\label{equ_foc}
\end{IEEEeqnarray}
The solution is given by the $W\geq \frac{m^{\A}}{2}$ cases in (\ref{equ_x1}) and (\ref{equ_x2}), in which $x_1^{\AB}$ will be negative if $W< \frac{m^{\A}}{2}$. Then we can guess that both SPs may abandon market $\AB$ (i.e., $x_1^{\AB}=x_2^{\AB}=0$) when the bandwidth $W$ is not large enough. But when the bandwidth is larger than $\frac{m^{\A}}{2}$, the solution to (\ref{equ_foc}) might be an equilibrium. Thus we will consider two cases, namely, $W\geq \frac{m^{\A}}{2}$ and $W < \frac{m^{\A}}{2}$, and prove the uniqueness of equilibrium for each case.

{\it Case $W\geq \frac{m^{\A}}{2}$}: We first show that the equilibrium is unique by showing it is a potential game and the potential function is strictly concave. Then we show that the solution to (\ref{equ_foc}) is feasible and thus is indeed the unique equilibrium. 

Let $\mathbf{x}_1=\left[x_1^{\A}, x_1^{\AB}\right]^T$ and $\mathbf{x}_2=\left[x_2^{\B}, x_2^{\AB}\right]^T$. Then we can construct a potential function from revenue (\ref{equ_R1}) and (\ref{equ_R2}) as follows:

\begin{IEEEeqnarray}{cl}
&\Phi(\mathbf{x}_1, \mathbf{x}_2) =  \nonumber\\
&- \Biggl[\left(\frac{1}{\mA}+\frac{1}{W}\right){x_1^{\A}}^2 + \left(\frac{1}{\mAB}+\frac{1}{W}\right){x_1^{\AB}}^2\nonumber\\
&+ \left(\frac{1}{\mAB}+\frac{1}{W}\right){x_2^{\AB}}^2 + \left(\frac{1}{\mB}+\frac{1}{W}\right){x_2^{\B}}^2\nonumber\\
&+\left(\frac{1}{\mAB}+\frac{1}{W}\right)x_1^{\AB}x_2^{\AB} + \frac{2}{W}x_1^{\A}x_1^{\AB} \nonumber\\
&+\frac{1}{W}x_1^{\A}x_2^{\AB} + \frac{1}{W}x_2^{\B}x_1^{\AB}+\frac{2}{W}x_2^{\B}x_2^{\AB}\Biggl]\nonumber\\
&+ x_1^{\A} + x_2^{\B} + x_1^{\AB} + x_2^{\AB}.
   \label{equ_potential}
\end{IEEEeqnarray}
One can verify that 
\begin{IEEEeqnarray}{c}
\Phi(\mathbf{x}^\prime_1, \mathbf{x}_2)-\Phi(\mathbf{x}_1, \mathbf{x}_2) = R_1(\mathbf{x}^\prime_1, \mathbf{x}_2) - R_1(\mathbf{x}_1, \mathbf{x}_2), \\
\Phi(\mathbf{x}_1, \mathbf{x}^\prime_2)-\Phi(\mathbf{x}_1, \mathbf{x}_2) = R_2(\mathbf{x}_1, \mathbf{x}^\prime_2) - R_2(\mathbf{x}_1, \mathbf{x}_2).
\label{equ_potential_def}
\end{IEEEeqnarray}

We can rewrite (\ref{equ_potential}) in a quadratic form 
\begin{IEEEeqnarray}{c}
\Phi(\mathbf{x}) = -\mathbf{x}^T\mathbf{A}\mathbf{x} + \mathbf{1}^T\mathbf{x},
\label{equ_potential_quadra}
\end{IEEEeqnarray}
where 
\begin{IEEEeqnarray}{l}
\hspace{70pt}\mathbf{A}=\nonumber\\
\begin{bmatrix}
\frac{1}{\mA}+\frac{1}{W} & \frac{1}{W} & \frac{1}{2W} & 0\\
\frac{1}{W} & \frac{1}{\mAB}+\frac{1}{W} & \frac{1}{2\mAB}+\frac{1}{2W} & \frac{1}{2W} \\
\frac{1}{2W} & \frac{1}{2\mAB}+\frac{1}{2W} & \frac{1}{\mAB}+\frac{1}{W} & \frac{1}{W} \\
0 & \frac{1}{2W} & \frac{1}{W} & \frac{1}{\mB}+\frac{1}{W}
\end{bmatrix},\nonumber	\\
\label{equ_potential_A}
\end{IEEEeqnarray}
and 
\begin{IEEEeqnarray}{c}
\mathbf{x}=\left[x_1^{\A} ~ x_1^{\AB} ~ x_2^{\AB} ~ x_2^{\B}\right]^T.
\end{IEEEeqnarray}
Next, we prove the uniqueness by showing that $\mathbf{A}$ is positive definite given $m^{\A}=m^{\B}$, i.e., $\mathbf{y}^T\mathbf{A}\mathbf{y}>0$, $\forall \mathbf{y}\in\mathbb{R}^4$.

From (\ref{equ_potential_A}), we observe that $\mathbf{y}^T\mathbf{A}\mathbf{y}$ is a decreasing function of $m^{\A}$ as $1/m^{\A}$ only appears in the diagonal. We can also show that it is a decreasing function of $m^{\AB}$ as 
\begin{IEEEeqnarray}{c}
\begin{bmatrix}
y_2 & y_3
\end{bmatrix}
\begin{bmatrix}
\frac{1}{\mAB}+\frac{1}{W} & \frac{1}{2\mAB}+\frac{1}{2W} \\
\frac{1}{2\mAB}+\frac{1}{2W} & \frac{1}{\mAB}+\frac{1}{W}
\end{bmatrix}
\begin{bmatrix}
y_2 \\ y_3
\end{bmatrix}\nonumber\\
= \left(\frac{1}{\mAB}+\frac{1}{W} \right)\frac{(y_2+y_3)^2+y_2^2+y_3^2}{2}.
\end{IEEEeqnarray}
Thus we can bound $\mathbf{y}^T\mathbf{A}\mathbf{y}$ from below by substituting in the maximum values of $m^{\A}$ and $m^{\AB}$:
\begin{IEEEeqnarray}{cl}
\mathbf{y}^T\mathbf{A}\mathbf{y} &\ge \mathbf{y}^T\mathbf{A}\mathbf{y}\big|_{m^{\A}=2W, ~m^{\AB}=1}\\
&=\frac{1}{W}\Bigg(\frac{3}{2}y_1^2+2y_1y_2+y_1y_3\nonumber\\
&+\frac{3}{2}y_4^2+2y_3y_4+y_2y_4\nonumber\\
&+(W+1)(y_2^2+y_3^2+y_2y_3)\Bigg).\label{equ_yAy_bound}
\end{IEEEeqnarray}
The lower bound (\ref{equ_yAy_bound}) is a quadratic function of $y_1$ and $y_4$, whose minimizer is given by $y_1^\star=-(2y_2+y_3)/3$ and $y_4^\star=-(2y_3+y_2)/3$. Thus by plugging in $y_1^\star$ and $y_4^\star$, we have
\begin{IEEEeqnarray}{cl}
\mathbf{y}^T\mathbf{A}\mathbf{y} &\ge \frac{1}{6W}(y_2-y_3)^2+\frac{1}{2}\left((y_2+y_3)^2+ y_2^2+y_3^2\right)\\
&\ge 0,
\end{IEEEeqnarray}
where the equality holds only when $y_2=y_3=0$ leading to $y_1^\star=y_4^\star=0$. Thus we can conclude that $\mathbf{y}^T\mathbf{A}\mathbf{y}>0$ for all $\mathbf{y}\neq 0$, which by definition proves $\mathbf{A}$ is positive definite. Therefore, the potential function (\ref{equ_potential}) is strictly concave, and thus the equilibrium is unique. The existence of an equilibrium is obvious as the set of $\mathbf{x}$ is compact because of (\ref{equ_x_dom1})--(\ref{equ_x_dom3}). 

To show (\ref{equ_x1}) and (\ref{equ_x2}) are indeed the solution, we only need to prove that they are feasible, i.e., $x_1^{\A}\leq \mA$ and $x_1^{\AB}\leq \mAB/2$. The first inequality is obvious since $W/C\leq 1$. For the second inequality, we substitute $m^{\AB}$ with $1-2m^{\A}$ and then it can be reduced to 
\begin{IEEEeqnarray}{c}
2W^2+(6-4m^{\A})W+6\left({m^{\A}}\right)^2-3m^{\A} \ge 0,
\label{equ_2ndreduced}
\end{IEEEeqnarray}
where the LHS is a quadratic function of $W$, with $W=m^{\A}-3/2<0$ as the minimizer. Recall that $W\ge \mA/2$, thus the actual minimizer is $W^\star = \mA/2$. Setting $W=W^\star$ in (\ref{equ_2ndreduced}), we have
\begin{IEEEeqnarray}{c}
\frac{9}{2}\left({\mA}\right)^2 \ge 0,
\end{IEEEeqnarray}
which always holds. Therefore, the solution in (\ref{equ_x1}) and (\ref{equ_x2}) is feasible and thus is the equilibrium.

{\it Case $W< \frac{m^{\A}}{2}$}: The potential function (\ref{equ_potential_quadra}) is not necessarily concave with $W< \frac{m^{\A}}{2}$ so we will adopt a different approach. The equilibrium(s) is always given by the solution to the following problem
\begin{IEEEeqnarray}{c}
   \underset{\mathbf{x}\ge 0}{\arg\max}~\Phi(\mathbf{x}) \\
   \text{s.t.  (\ref{equ_x_dom1})--(\ref{equ_x_dom3})}.\nonumber
\end{IEEEeqnarray}
And, as mentioned before, the first-order solution (i.e., the solution to $\nabla\Phi(\mathbf{x})=0$, which is equivalent to (\ref{equ_foc})) is not feasible. Thus the equilibrium(s) must lie on a boundary\footnote{Again, the existence of an equilibrium follows from the compact set formed by (\ref{equ_x_dom1})--(\ref{equ_x_dom3}).}. Next, we will show the uniqueness by excluding all boundaries except one.

First, it is easy to rule out the boundaries on market-size constraints (i.e., $x_1^{\A}=\mA$, $x_1^{\AB} + x_2^{\AB}= \mAB$, and $x_2^{\B}= \mB$). For example, with $x_1^{\A}=\mA$, the delivered price $p_d^{\A}$ is zero. This results in a non-positive service price $p_d^{\A}- l^{\A}$ charged by SP1 in market $\A$. Thus, SP1 can be better off if it stops serving users in this market (i.e., by letting $x_1^{\A}=0$) as it is currently paying these users for using its service. 

Second, to rule out the boundaries on non-negative constraints (i.e., $x_1^{\A}=0$, $x_1^{\AB}=0$, $x_2^{\AB}= 0$, and $x_2^{\B}=0$), we first prove the following lemmas which will help us check all combinations of these boundary conditions.
\begin{lemma}
\label{lemma_1}
Any $\mathbf{x}$ such that $x_1^{\A}=0$ but $x_1^{\AB}>0$ (or $x_2^{\B}=0$ but $x_2^{\AB}>0$) is not an equilibrium.
\end{lemma}
The intuition is that it is impossible to have an equilibrium in which an SP is willing to serve users in a more crowded market rather than a less crowded one.
\begin{proof}
    With $x_1^{\A}=0$ and $x_1^{\AB}>0$, we have 
    \begin{IEEEeqnarray}{cl}
        &p_d^{\A}- l^{\A}\nonumber\\
        =&1-\frac{x_1^{\AB}+x_2^{\AB}}{W}\nonumber\\
        >&1-\frac{x_1^{\AB}+x_2^{\AB}}{\mAB}-\frac{x_1^{\AB}+x_2^{\AB}+x_2^{\B}}{W}\nonumber\\
        =&p_d^{\AB}- l^{\AB},\label{equ_ineq_service_price}
    \end{IEEEeqnarray}
    which shows that the service price in market $\A$ is higher than that in market $\AB$.

    Consider a deviation $\mathbf{x}^\prime_1=\left[\Delta, x_1^{\AB}-\Delta\right]^T$ from $\mathbf{x}_1=\left[0, x_1^{\AB}\right]^T$ with $\Delta>0$. The difference in SP1's revenue for any $\mathbf{x}_2$ is
    \begin{IEEEeqnarray}{cl}
        &R_1(\mathbf{x}^\prime_1,\mathbf{x}_2)-R_1(\mathbf{x}_1,\mathbf{x}_2)\nonumber\\
        = &\Delta\Bigg[\left(p_d^{\A}- l^{\A} - \frac{\Delta}{\mA}\right)\nonumber\\
        & - \left(p_d^{\AB}- l^{\AB} + \frac{\Delta}{\mAB} \right) + \frac{x_1^{\AB}}{\mAB}\Bigg].
    \end{IEEEeqnarray}
    Given (\ref{equ_ineq_service_price}), we can find always find a small enough $\Delta$ such that $R_1(\mathbf{x}^\prime_1,\mathbf{x}_2)-R_1(\mathbf{x}_1,\mathbf{x}_2)>0$, indicating $\mathbf{x}^\prime_1$ is a profitable deviation from $\mathbf{x}_1$ and thus it is not an equilibrium.
    
\end{proof}

\begin{lemma}
\label{lemma_2}
Any $\mathbf{x}$ such that $x_1^{\A}=0$ and $x_1^{\AB}=0$ (or $x_2^{\B}=0$ but $x_2^{\AB}=0$) is not an equilibrium.
\end{lemma}
This lemma suggests that an SP cannot completely squeeze another SP out of the entire market.
\begin{proof}
    Consider a deviation $\mathbf{x}^\prime_1=\left[\Delta, 0\right]^T$ from $\mathbf{x}_1=\left[0, 0\right]^T$ with $\Delta>0$. For any $\mathbf{x}_2$, the revenue of SP1 is 
    \begin{IEEEeqnarray}{cl}
        &R_1(\mathbf{x}^\prime_1,\mathbf{x}_2)\\
        =&\Delta\left(1-\frac{\Delta}{\mA}-\frac{\Delta+x_2^{\AB}}{W}\right)\\
        =&\Delta\left(1-\frac{x_2^{\AB}}{W}-\left(\frac{1}{\mA}+\frac{1}{W}\right)\Delta\right).
    \end{IEEEeqnarray}
    Note that $1-\frac{x_2^{\AB}}{W}>0$ always holds for any $\mathbf{x}_2$ at an equilibrium, otherwise SP2 would incur a negative service price in both market $\AB$ and $\B$. Thus there always exists a small enough $\Delta$ such that $R_1(\mathbf{x}^\prime_1,\mathbf{x}_2)>0$, indicating $\mathbf{x}^\prime_1$ is a profitable deviation from $\mathbf{x}_1$ and thus it is not an equilibrium.
\end{proof}

\begin{lemma}
\label{lemma_3}
Any $\mathbf{x}$ such that $x_1^{\AB}=0$ but $x_2^{\AB}>0$ (or $x_2^{\AB}=0$ but $x_1^{\AB}>0$) is not an equilibrium.
\end{lemma}
This lemma suggests that an SP cannot completely squeeze another SP out of market $\AB$.
\begin{proof}
    Consider a strategy profile $\mathbf{x}=[x_1^{\A}, 0, x_2^{\AB},  x_2^{\B}]^T$ with $x_2^{\AB}>0$. If it is an equilibrium, we must have $x_1^{\A}>0$ by Lemma \ref{lemma_2} and $x_2^{\B}>0$ by Lemma \ref{lemma_1}. Thus, to be an equilibrium, $\mathbf{x}$ needs to be the solution to the following equations
    \begin{IEEEeqnarray}{c}
   \left[\frac{\partial R_1|_{x_1^{\AB}=0}}{\partial x_1^{\A}} ~\frac{\partial R_2|_{x_1^{\AB}=0}}{\partial x_2^{\AB}}~\frac{\partial R_2|_{x_1^{\AB}=0}}{\partial x_2^{\B}}\right]^T
    =\mathbf{0}. \label{equ_lemma3_equations}
\end{IEEEeqnarray}
One can verify that the solution always has $x_1^{\A}x_2^{\AB}<0$, indicating there is no feasible solution to (\ref{equ_lemma3_equations}).
\end{proof}

With Lemma \ref{lemma_1}--\ref{lemma_3}, we can check all 15 combinations of boundary conditions.\footnote{There are 4 $x$'s with total ${4\choose 1}+{4\choose 2}+{4\choose 3}+{4\choose 4}=15$ different combinations of boundary conditions.} For example, $x_1^{\A}=0$ and $x_1^{\AB},x_2^{\AB},x_2^{\B}>0$ can be ruled out by Lemma \ref{lemma_1}. The only boundary condition not eliminated is where $x_1^{\A}, x_2^{\B}>0$ and $x_1^{\AB}, x_2^{\AB}=0$. $x_1^{\A}$ and $ x_2^{\B}$ are the solution to the following equations 
\begin{IEEEeqnarray}{c}
\left[\frac{\partial R_1|_{x_1^{\AB}=x_2^{\AB}=0}}{\partial x_1^{\A}} ~\frac{\partial R_2|_{x_1^{\AB}=x_2^{\AB}=0}}{\partial x_2^{\B}}\right]^T
=\mathbf{0}.
\end{IEEEeqnarray}
The solution is given by the $0 \leq W < \frac{m^{\A}}{2}$ case in (\ref{equ_x1}) which is unique.

\end{proof}

\section{Numerical Results}

\subsection{Equilibrium Quantities}

\begin{figure}[tb]
    \centering
    \begin{subfigure}[b]{0.45\linewidth}
        \centering
        \includegraphics[width=\linewidth]{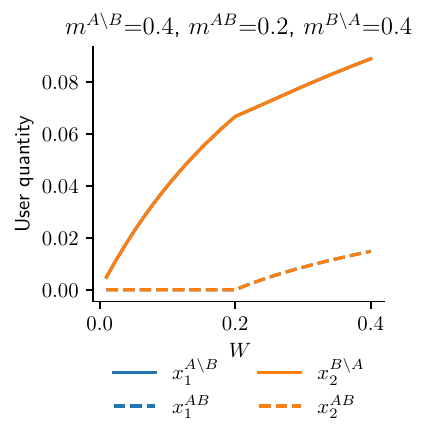}        
        \caption{Small market $\AB$}
        \label{fig_equilibrium_1}
    \end{subfigure}
    \hfill
    \begin{subfigure}[b]{0.45\linewidth}
        \centering
        \includegraphics[width=\linewidth]{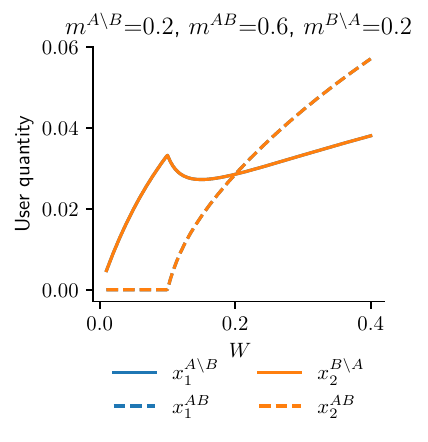}
        \caption{Large market $\AB$}
        \label{fig_equilibrium_2}
    \end{subfigure}
    \caption{User quantities at equilibrium versus bandwidth $W$ for two symmetric cases.}
    \label{fig_equilibrium_symm}
\end{figure}

\begin{figure}[tb]
    \centering
    \begin{subfigure}[b]{0.45\linewidth}
        \centering
        \includegraphics[width=\linewidth]{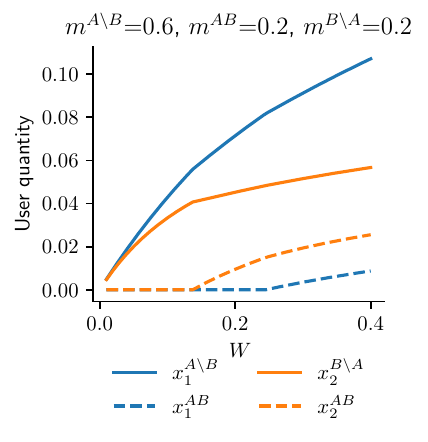}        
        \caption{Small market $\AB$}
        \label{fig_equilibrium_3}
    \end{subfigure}
    \hfill
    \begin{subfigure}[b]{0.45\linewidth}
        \centering
        \includegraphics[width=\linewidth]{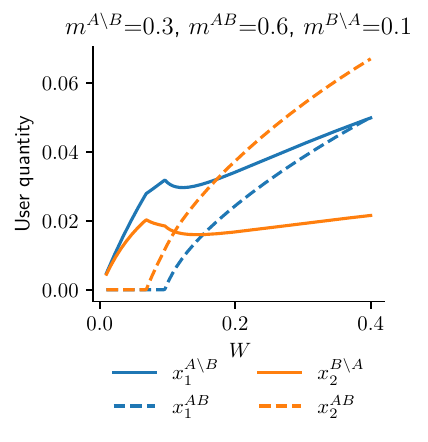}
        \caption{Large market $\AB$}
        \label{fig_equilibrium_4}
    \end{subfigure}
    \caption{User quantities at equilibrium versus bandwidth $W$ for two asymmetric cases.}
    \label{fig_equilibrium_asymm}
\end{figure}


In Fig.~\ref{fig_equilibrium_symm} we illustrate the unique equilibrium user quantities versus the bandwidth $W$ for two symmetric settings (i.e. when $m^{\A}= m^{\B}$) with different market size values $\mAB$. Note due to symmetry in these cases the user quantities for each SP are the same (e.g., $x^{\A} = x^{\B}$). As suggested in (\ref{equ_x2}), both SPs do not enter market $\AB$ until $W$ is greater than $\frac{m^{\A}}{2}$, which is $0.2$ and $0.1$ in Fig.~\ref{fig_equilibrium_1} and Fig.~\ref{fig_equilibrium_2}, respectively. The intuition is that SPs avoid overlapping areas to minimize latency costs when bandwidth is limited. However, even with this preference, SPs will still begin serving users in the overlapping area if bandwidth $W$ further increases. Initially, as $W$ increases the SPs serve more users on the non-overlapping sub-markets.  This decreases the delivered price in these areas and thus the marginal revenue gained by adding users, until at some point it is more attractive to add users on the overlapping area (with a higher delivered price).  Note also that when $\AB$ is larger as in 
Fig.~\ref{fig_equilibrium_2}, this increase in users in the overlapping area results in an initial {\it decrease} in users in the non-overlapping area.  This is because increasing traffic in the overlapping area increases the latency cost of SPs on both their overlapping and non-overlapping bands. This in turn reduces the marginal benefit of serving users on non-overlapping band.

Fig.~\ref{fig_equilibrium_asymm} shows the equilibrium user quantities versus $W$ for two asymmetric cases ($\mA>\mB$) with different market size values $\mAB$. Note that SP2 has a smaller dedicated market compared to SP1. As in the symmetric case, when the bandwidth is limited, no SP joins the overlapping market. Interestingly, as $W$ keeps increasing, SP2 first enters the overlapping market, and then SP1 enters later.  From SP1's perspective, since it has a relatively larger dedicated market, it does not benefit from competing with SP2 in market $\AB$ if $W$ is not large enough. However, as $W$ keeps increasing, SP1 will  eventually enter the overlapping market to compete with SP2 directly, but will always serve few customers in that band than SP1 does.  For large $\mAB$, as shown in Fig.~\ref{fig_equilibrium_4}, the SPs once again may decrease the users served in the non-overlapping market, when they increase those served in the overlapping market.  Note here that the decrease for each SP begins at the value of $W$ when that SP begins serving customers in the overlapping market.

Fig.~\ref{fig_equilibrium_symm} and Fig.~\ref{fig_equilibrium_asymm} also suggest that the smaller the overlapping market, the higher the bandwidth required for either SP to join the overlapping market. This observation provides insight for a regulator, showing that providing sufficient bandwidth is the key to encouraging SPs to serve customers in overlapping markets.

\begin{figure}[tb]
    \centering
    \begin{subfigure}[b]{0.45\linewidth}
        \centering
        \includegraphics[width=\linewidth]{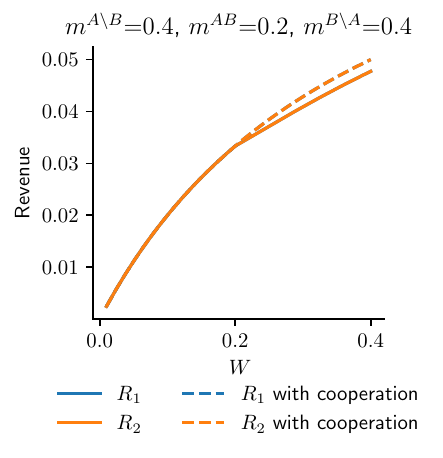}
        \caption{Small market $\AB$}
        \label{fig_revenue_1}
    \end{subfigure}
    \hfill
    \begin{subfigure}[b]{0.45\linewidth}
        \centering
        \includegraphics[width=\linewidth]{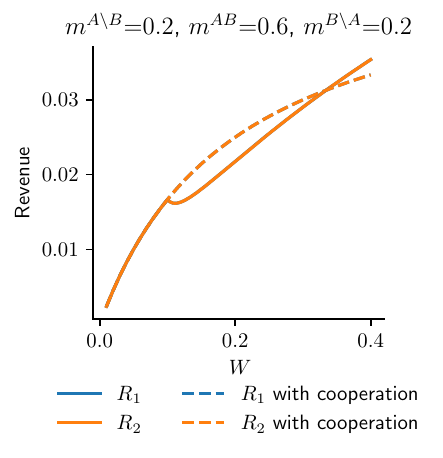}
        \caption{Large market $\AB$}
        \label{fig_revenue_2}
    \end{subfigure}
    \caption{SP's revenue versus bandwidth $W$ for two symmetric cases.}
    \label{fig_revenue_symm}
\end{figure}

\begin{figure}[tb]
    \centering
    \begin{subfigure}[b]{0.45\linewidth}
        \centering
        \includegraphics[width=\linewidth]{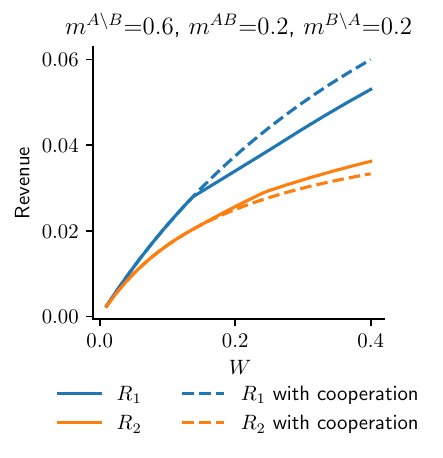}
        \caption{Small market $\AB$}
        \label{fig_revenue_3}
    \end{subfigure}
    \hfill
    \begin{subfigure}[b]{0.45\linewidth}
        \centering
        \includegraphics[width=\linewidth]{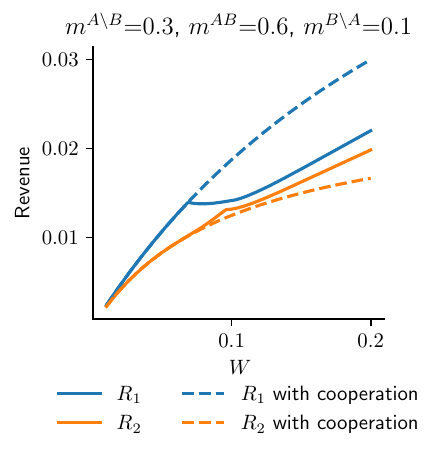}
        \caption{Large market $\AB$}
        \label{fig_revenue_4}
    \end{subfigure}
    \caption{SP's revenue versus bandwidth $W$ for two asymmetric cases.}
    \label{fig_revenue_asymm}
\end{figure}

\subsection{SP's Revenue}
Next, we show some examples of the SP's revenue versus $W$ in
Fig.~\ref{fig_revenue_symm} and Fig.~\ref{fig_revenue_asymm} for the same settings as in Fig.~\ref{fig_equilibrium_symm} and Fig.~\ref{fig_equilibrium_asymm}, respectively.  On these figures we also include the revenue obtained when the SPs ``cooperate" and do not serve any customers in the overlapping band for any value of $W$.\footnote{This is one simple example of how SPs could cooperate.  We leave the consideration of other approaches for future work.}

For a small overlapping market, the equilibrium revenue is an increasing function of $W$, as shown in Fig.~\ref{fig_revenue_1} and Fig.~\ref{fig_revenue_3}. For large enough $\mAB$, however, revenue may decrease shortly after either SP enters the overlapping market, as shown in Fig.~\ref{fig_revenue_2} and Fig.~\ref{fig_revenue_4}. This is because when one SP uses the overlapping market, it increases the latency cost for the other SP on its non-overlapping band and does not account for this externality when determining its quantity.
This can be viewed as a type of Braess's paradox where adding more resources (i.e., increasing $W$) leads to lower revenue.  Similar effects have been noted in model of markets without geographic separation (e.g.~\cite{nguyen2016cost}). The results here show that such behavior depends on the amount of geographic separation.

 As mentioned earlier, SPs have a preference for their dedicated markets when bandwidth is limited due to the higher latency cost in the overlapping area. When  SPs enter the overlapping market, by comparing their equilibrium revenue to the revenue obtained in the cooperation case, we can see that entering this market actually reduces both SPs revenue in the symmetric case for a range of $W$ (when $W$ is large enough both SPs would benefit from entering the overlapping market).\footnote{This is shown in Fig.~\ref{fig_revenue_2}; it also occurs for the scenario in Fig.~\ref{fig_revenue_1} if $W$ is large enough, though the range of $W$ needed is not shown.}
   This suggests in the symmetric setting the SPs may have an incentive to enter into an agreement to not serve customers in the overlapping region for some rang of $W$.
  For asymmetric markets as in Fig.~\ref{fig_revenue_asymm}, only the SP with a larger dedicated market (SP1) benefits from cooperation, while SP2's revenue decreases.  However, the total revenue increases for a range of $W$ under cooperation so that it would be profitable for SP1 to compensate SP2 for cooperating. Again, this only holds for a range of $W$ and if $W$ is large enough (not shown), the total revenue would increase if both SPs enter the overlapping market.

\subsection{Consumer Surplus}

\begin{figure}[tb]
    \centering   
    \begin{subfigure}[b]{0.45\linewidth}
        \centering
            \includegraphics[height=4cm]
        {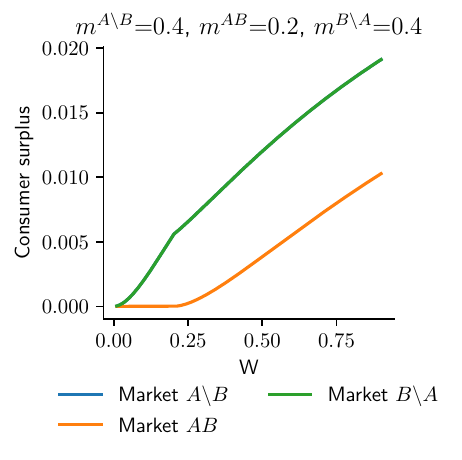}
        \caption{Sub-markets}
        \label{fig_cs_1}
    \end{subfigure}
    \hfill
    \begin{subfigure}[b]{0.45\linewidth}
        \centering
        \includegraphics[height=4cm]
            {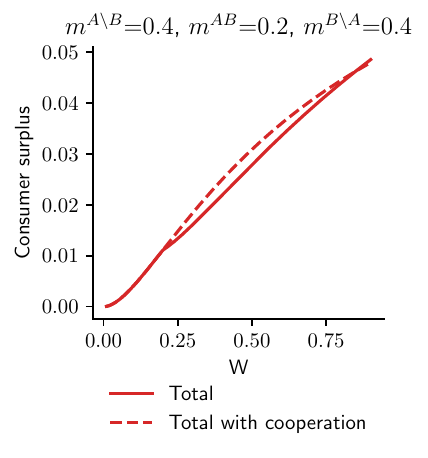}
        \caption{Whole market}
        \label{fig_cs_2}
    \end{subfigure}
    \caption{Consumer surplus versus bandwidth $W$ (symmetric case with small $\mAB$).}
    \label{fig_cs_symm_small}
\end{figure}

\begin{figure}[tb]
    \centering   
    \begin{subfigure}[b]{0.45\linewidth}
        \centering
            \includegraphics[height=4cm]
        {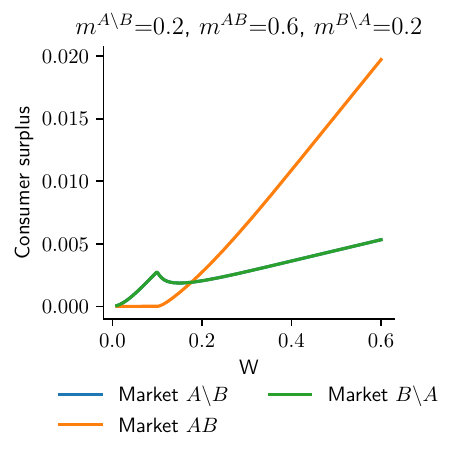}
        \caption{Sub-markets}
        \label{fig_cs_3}
    \end{subfigure}
    \hfill
    \begin{subfigure}[b]{0.45\linewidth}
        \centering
        \includegraphics[height=4cm]
            {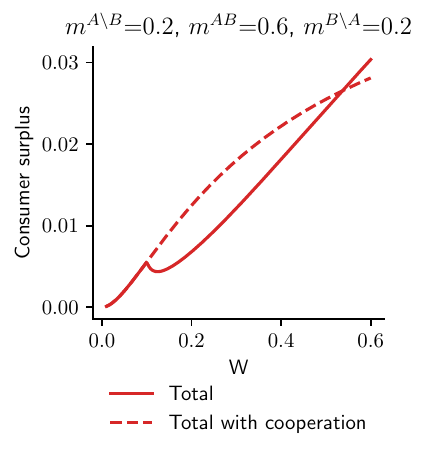}
        \caption{Whole market}
        \label{fig_cs_4}
    \end{subfigure}
    \caption{Consumer surplus versus bandwidth $W$ (symmetric case with large $\mAB$).}
    \label{fig_cs_symm_large}
\end{figure}

Next we consider the consumer surplus versus $W$ for the same set of scenarios. 
Fig.~\ref{fig_cs_symm_small} and Fig.~\ref{fig_cs_symm_large} illustrate how consumer surplus varies with $W$ for the two symmetric cases with small and large value of 
$\mAB$, respectively.  Figures \ref{fig_cs_1} and \ref{fig_cs_3}, show the consumer surplus for each sub-market, while Fig.~\ref{fig_cs_2} and Fig.~\ref{fig_cs_4} show the surplus for the entire market as well as the surplus obtained by the ``cooperation" case discussed previously. 

Focusing first on the individual sub-markets, 
the consumer surplus of market $\AB$ is zero until SPs enter this market when the bandwidth $W$ is large enough. Notice that after SPs start serving users in the overlapping market, the rate of increase of consumer surplus in their dedicated market slows down when $\mAB$ is small (Fig.~\ref{fig_cs_1}) or even decreases when $\mAB$ is large (Fig.~\ref{fig_cs_3}). Turning to the surplus for the entire market it can be seen that when $\mAB$ is small, this is increasing in $W$ (Fig.~\ref{fig_cs_2}), while when $\mAB$ is large (Fig.~\ref{fig_cs_4}), the consumer surplus can decrease when $W$ increases past the point where both SPs enter the overlapping market. Apparently, in this case the decrease in surplus on the non-overlapping markets is greater than the surplus gained on the overlapping markets.

When the SPs enter the overlapping market, they compete for customers in that market.  Surprisingly, as shown in Fig.~\ref{fig_cs_2} and Fig.~\ref{fig_cs_4}, such competition does not always increase overall consumer surplus compared to the case where the SPs cooperate and stay out of this market. While this competition can benefit users in the overlapping market, it reduces the surplus of customers in the non-overlapping markets compared to the cooperative case. Recall, as shown in Fig.~\ref{fig_revenue_symm}, the SPs can also improve their revenue by cooperating in this way.  Thus, cooperatively avoiding the overlapping market can improve both the SPs' revenue and the consumer surplus for a range of $W$. However, it is important to note that this increase in total consumer surplus comes at the expense of consumer surplus in the overlapping area, where no one is served. Consequently, regulators might consider subsidizing users in the overlapping area, using the benefits derived from such cooperation. Note also that, as shown in Fig.~\ref{fig_cs_2} and Fig.~\ref{fig_cs_4}, with sufficiently large $W$ cooperation results in lower surplus compared to the case where both SPs compete in the overlapping region. 

\begin{figure}[tb]
    \centering   
    \begin{subfigure}[b]{0.45\linewidth}
        \centering
            \includegraphics[height=4cm]
        {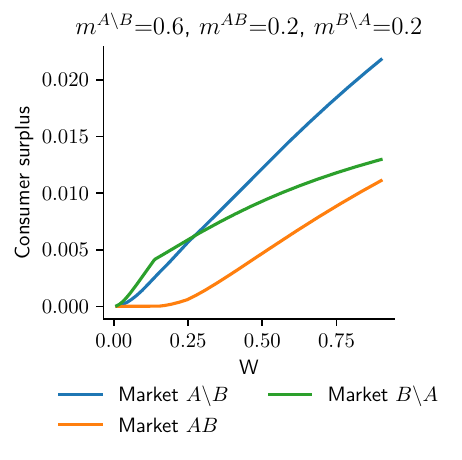}
        \caption{Sub-markets}
        \label{fig_cs_5}
    \end{subfigure}
    \hfill
    \begin{subfigure}[b]{0.45\linewidth}
        \centering
        \includegraphics[height=4cm]
            {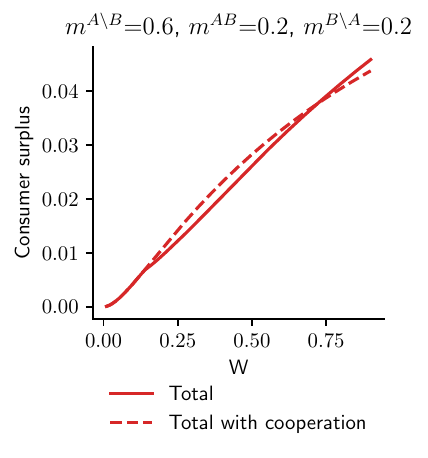}
        \caption{Whole market}
        \label{fig_cs_6}
    \end{subfigure}
    \caption{Consumer surplus versus bandwidth $W$ (asymmetric case with small $\mAB$).}
    \label{fig_cs_asymm_small}
\end{figure}

\begin{figure}[tb]
    \centering   
    \begin{subfigure}[b]{0.45\linewidth}
        \centering
            \includegraphics[height=4cm]
        {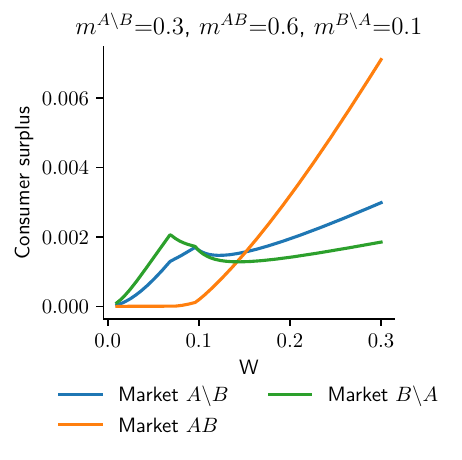}
        \caption{Sub-markets}
        \label{fig_cs_7}
    \end{subfigure}
    \hfill
    \begin{subfigure}[b]{0.45\linewidth}
        \centering
        \includegraphics[height=4cm]
            {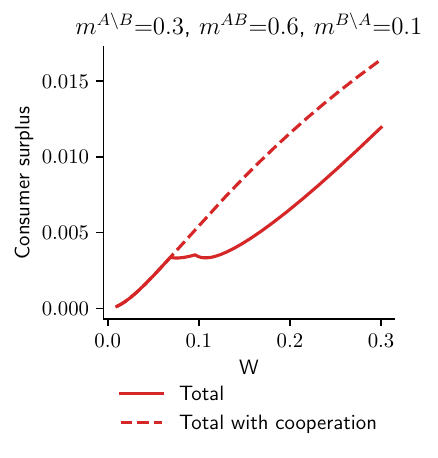}
        \caption{Whole market}
        \label{fig_cs_8}
    \end{subfigure}
    \caption{Consumer surplus versus bandwidth $W$ (asymmetric case with large $\mAB$).}
    \label{fig_cs_asymm_large}
\end{figure}

In Fig.~\ref{fig_cs_asymm_small} and Fig.~\ref{fig_cs_asymm_large} we show similar plots of consumer surplus for asymmetric scenarios with small and large values of $\mAB$, respectively. In both cases, for small values of $W$, there is no surplus generated in the overlapping market as the SPs do not compete in that market.  Also, note that for small values of $W$, the smaller SP (SP2) creates more surplus even though it is serving fewer customers.  This is due to the different demand curves in the two markets. SP2's demand curve in $\mB$ has a steeper slope compared to SP1's demand curve in $\mA$, meaning that it has to charge a lower delivered price to serve a similar number of customers, which in turn leads to larger welfare. 

In these asymmetric models, When $\mAB$ is small, consumer surplus increases as a function of bandwidth $W$. However, with a large $\mAB$ (Fig.~\ref{fig_cs_asymm_large}), consumer surplus is decreasing in $W$ around the values where the SPs enter the overlapping market. Once again, if the SPs cooperate and do not enter the overlapping market, this can increase the overall surplus for a range of $W$.

\subsection{Social Welfare}
\begin{figure}[tb]
    \centering   
    \begin{subfigure}[b]{0.45\linewidth}
        \centering
            \includegraphics[height=4cm]
        {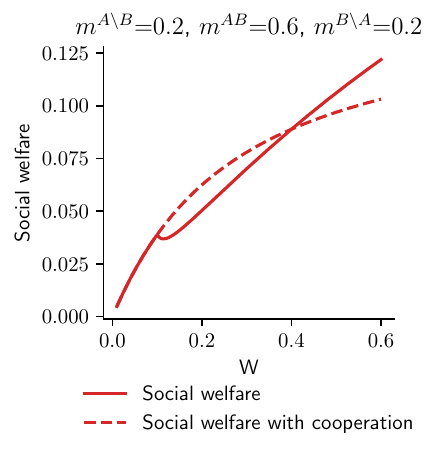}
        \caption{$\mA = \mB$}
        \label{fig_sw_1}
    \end{subfigure}
    \hfill
    \begin{subfigure}[b]{0.45\linewidth}
        \centering
        \includegraphics[height=4cm]
            {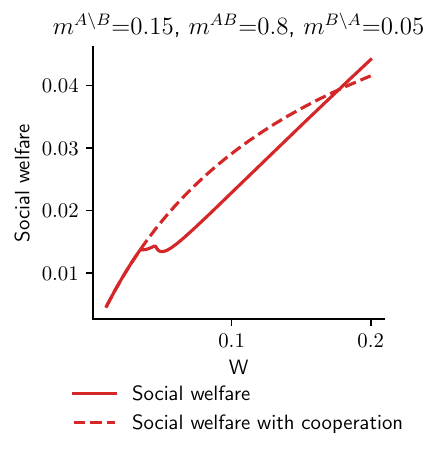}
        \caption{$\mA > \mB$}
        \label{fig_sw_2}
    \end{subfigure}
    \caption{Social welfare versus bandwidth $W$.}
    \label{fig_sw}
\end{figure}
Social welfare is defined as the sum of consumer surplus and the revenue of SPs. Recall that, when $\mAB$ is small, both consumer surplus and revenue increase as a function of bandwidth $W$, leading to a corresponding increase in social welfare. However, for large values of $\mAB$, increasing bandwidth in the market may not necessarily lead to higher social welfare. We illustrate this in Fig.~\ref{fig_sw} for the two large $\mAB$ scenerios. These show that 
 social welfare is not a monotonically increasing function of $W$. This is expected as we have already shown that both total revenue and consumer surplus may decrease in $W$ when the SPs first enter the overlapping market. We also show the welfare obtained when the SPs cooperate and do not enter this market, which yields a welfare improvement for a range of $W$.

\section{Conclusions}
We presented a model of a spectrum sharing market with two geographically separated SPs that have partially overlapping coverage areas. In this model, we proved that a unique Nash equilibrium always exists, where SPs avoid entering the overlapping market when bandwidth is limited. With sufficient bandwidth, SPs will enter the overlapping market.  However, the resulting revenue and consumer surplus may decrease as a function of the amount of bandwidth once the SPs enter the overlapping market.  We also showed that the revenue and consumer surplus can both be improved for a range of bandwidth values by allowing the SPs to cooperate and not enter this overlapping market. This suggests that such cooperation may be desirable, but would need to be balanced by considering the fairness to the users within this overlapping market.

\bibliographystyle{IEEEtran}
\bibliography{conference_101719}
\end{document}